\documentclass[journal,10pt]{IEEEtran} 

\usepackage[T1]{fontenc}

\usepackage{graphicx}
\usepackage[export]{adjustbox}
\usepackage[dvipsnames]{xcolor}
\usepackage{bbm}
\usepackage{kotex} 

\graphicspath{{./Figs/}}
\usepackage{subfigure}
\usepackage{stfloats}
\ifCLASSINFOpdf
\else
\fi

\usepackage{amssymb}
\usepackage{amsmath}
\usepackage[cmintegrals]{newtxmath}
\usepackage{stackengine}
\def\delequal{\mathrel{\ensurestackMath{\stackon[1pt]{=}{\scriptstyle\Delta}}}}

\usepackage{algorithmic}
\usepackage{algorithm}

\usepackage{amsthm}
\usepackage{xpatch}

\newtheorem{thm}{Theorem}
\newtheorem{lem}{Lemma}

\xpatchcmd{\proof}{\hskip\labelsep}{\hskip5\labelsep}{}{}

\usepackage{mathtools}

\usepackage{cite}

\usepackage{array}

\usepackage{multirow}

\begin{document}
\title{Modeling and Analysis of Hybrid GEO-LEO Satellite Networks}

\author{Dong-Hyun Jung, Hongjae Nam, Junil Choi, and David J. Love\\
\vspace{-20pt}
\thanks{D.-H. Jung is with the Satellite Communication Research Division, Electronics and Telecommunications Research Institute, Daejeon, 34129, South Korea, and also with the Department of Information and Communication Engineering, University of Science and Technology, Daejeon, 34113, South Korea (e-mail: dhjung@etri.re.kr).}
\thanks{H. Nam and D. J. Love are with the Elmore Family School of Electrical and Computer Engineering, Purdue University, West Lafayette, IN 47907 USA (e-mail: nam86@purdue.edu; djlove@purdue.edu).}
\thanks{J. Choi is with the School of Electrical Engineering, KAIST, Daejeon, 34141, South Korea (e-mail: junil@kaist.ac.kr).}
}
\maketitle

\begin{abstract}
As the number of low Earth orbit (LEO) satellites rapidly increases, the consideration of frequency sharing or cooperation between geosynchronous Earth orbit (GEO) and LEO satellites is gaining attention. In this paper, we consider a hybrid GEO-LEO satellite network where GEO and LEO satellites are distributed according to independent Poisson point processes (PPPs) and share the same frequency resources. Based on the properties of PPPs, we first analyze satellite-visible probabilities, distance distributions, and association probabilities. Then, we derive an analytical expression for the network's coverage probability. Through Monte Carlo simulations, we verify the analytical results and demonstrate the impact of system parameters on coverage performance. The analytical results effectively estimate the coverage performance in scenarios where GEO and LEO satellites cooperate or share the same resource.

\end{abstract}

\IEEEpeerreviewmaketitle

\def\tcr{\textcolor{red}}
\def\tcb{\textcolor{blue}}
\def\tcg{\textcolor{green}}
\def\tco{\textcolor{orange}}

\def\E{\mathbb{E}}
\def\P{\mathbb{P}}
\def\T{\mathrm{T}}

\def\th{\mathrm{th}}
\def\max{\mathrm{max}}
\def\min{\mathrm{min}}
\def\out{\mathrm{out}}
\def\sec{\mathrm{sec}}
\def\cov{\mathrm{cov}}
\def\ml{\mathrm{ml}}
\def\sl{\mathrm{sl}}
\def\SINR{\mathsf{SINR}}
\def\SNR{\mathsf{SNR}}
\def\SIR{\mathsf{SIR}}

\def\re{r_{\mathrm{E}}}
\def\fc{f_{\mathrm{c}}}
\def\Pout{\mathcal{P}_{\out}}
\def\Pcov{\mathcal{P}_{\cov}}


\def\nI{n_{\mathrm{I}}}
\def\Npis{N_{\mathrm{pis}}}
\def\Npis{J}

\def\coeff{\omega}
\def\coeffG{\coeff^{\G}}
\def\coeffL{\coeff^{\L}}

\def\L{\mathrm{L}}
\def\G{\mathrm{G}}

\def\idxGL{\sigma} 
\def\asscidx{\tilde{\idxGL}}

\def\Loss{\bar{L}}

\def\IG{I_{\G | \asscidx}}
\def\IL{I_{\L | \asscidx}}
\def\IGG{I_{\G | \G}}
\def\IGL{I_{\G | \L}}
\def\ILG{I_{\L | \G}}
\def\ILL{I_{\L | \L}}

\def\Pt{P_{\mathrm{t}}}
\def\PtG{\Pt^{\G}}
\def\PtL{\Pt^{\L}}
\def\Pthat{\hat{P}_{\mathrm{t}}}

\def\Pr{P_{\mathrm{r}}}
\def\PrG{P_{\mathrm{r}}^{\G}}
\def\PrL{P_{\mathrm{r}}^{\L}}

\def\Gn{G_{n}}
\def\Go{G_{0}}
\def\Gohat{\hat{G}_{0}}
\def\Gtn{G_{\mathrm{t},n}}
\def\Gto{G_{\mathrm{t},0}}
\def\GoG{G_{0}^{\G}}
\def\GoL{G_{0}^{\L}}
\def\Gohats{\hat{G}_{0}^{\idxGL}}
\def\GohatG{\hat{G}_{0}^{\G}}
\def\GohatL{\hat{G}_{0}^{\L}}

\def\Gr{G_{\mathrm{r}}}

\def\BL{B^{\L}}
\def\BG{B^{\G}}
\def\Bhat{\hat{B}}

\def\alphaG{\alpha_{\G}}
\def\alphaL{\alpha_{\L}}
\def\alphahat{\hat{\alpha}}

\def\aL{a_{\L}}
\def\aG{a_{\G}}

\def\BPP{{\Phi}}  
\def\BPPL{{\Phi_{\L}}}  
\def\BPPG{{\Phi_{\G}}}  
\def\lG{\lambda_{\G}}
\def\lL{\lambda_{\L}}

\def\rminG{r_{\min}^{\G}(\phi)}
\def\rmaxG{r_{\max}^{\G}(\phi)}
\def\romaxG{r_{\mathrm{vis,max}}^{\G}}
\def\romaxL{r_{\mathrm{vis,max}}^{\L}}

\def\Pvis{\mathcal{P}_{\mathrm{vis}}}

\def\hnL{h_n^{\L}}
\def\hnG{h_n^{\G}}
\def\RoL{R_0^{\L}}
\def\RoG{R_0^{\G}}
\def\roLbiased{d_{\L}} 
\def\roGbiased{d_{\G}} 

\def\PsiL{\Lambda}
\def\PsiG{\Psi}

\def\NL{N_{\L}}
\def\RL{R^{\L}}
\def\AL{\mathcal{A}^\L} 

\def\NG{N_{\G}}
\def\RG{R^{\G}}
\def\AG{\mathcal{A}^\G} 
\def\As{\mathcal{A}^{\idxGL}} 

\def\ALrc{\AL(r)^c}
\def\Avis{\mathcal{A}_{\mathrm{vis}}}
\def\AvisG{\Avis^{\G}}
\def\AvisL{\Avis^{\L}}

\def\PcovG{\Pcov^{\G}}
\def\PcovL{\Pcov^{\L}}
\def\PcovGhat{\bar{\mathcal{P}}_{\mathrm{cov}}^{\G}}
\def\PcovLhat{\bar{\mathcal{P}}_{\mathrm{cov}}^{\L}}

\def\PasscG{\mathcal{P}_{\mathrm{assc}}^{\G}}
\def\PasscL{\mathcal{P}_{\mathrm{assc}}^{\L}}

\def\PsuccV{p_{\mathrm{vis}}}
\def\PsuccI{p_{\mathrm{int}}}

\def\AGvis{\AG_{\mathrm{vis}}}
\def\ALvis{\AL_{\mathrm{vis}}}
\def\Asvis{\As_{\mathrm{vis}}}

\def\phiinv{\phi_{\mathrm{inv}}}

\def\Pvoid{\PsuccV}




\vspace{-10pt}
\section{Introduction}\label{Sec:Intro}
\IEEEPARstart{S}{atellite}
communications can achieve global connectivity using their wide coverage capabilities. Since Release 15, the 3rd Generation Partnership Project (3GPP) has focused on integrating terrestrial networks (TNs) with non-terrestrial networks (NTNs) \cite{TR38.811,TR38.821}. By incorporating geosynchronous Earth orbit (GEO) satellites and low Earth orbit (LEO) satellites, communication services can be extended beyond the boundaries of terrestrial infrastructure. This technological advancement could enhance connectivity for aerial users, including drones, planes, and urban air mobility vehicles. To enable this, 3GPP is expected to establish a unified standard integrating TNs and NTNs starting with the 6G standard.

In 3GPP, the system-level performance of NTNs is evaluated by simulating a predefined hexagonal cell-based configuration on the UV plane, a two-dimensional coordinate system where 'U' and 'V' represent spatial coordinates with axes perpendicular to the satellite-Earth line \cite{my24APCCsub}. However, this approach requires excessive time to assess the impact of system parameters on performance.
Alternatively, the system-level performance of NTNs has recently been assessed through stochastic geometry, utilizing spatial point processes to effectively model satellite locations. Binomial point processes (BPPs) are commonly used to represent the distribution of LEO satellites, as the total number of satellites is finite~\cite{SGsatBPP20Otaki}. Additionally, the analysis based on Poisson point processes (PPP) has been validated in \cite{my22TCOM} using the Poisson limit theorem, which connects BPP- and PPP-based analyses. While many studies \cite{SGsatPPP23Park, SGsatBPP20TalgatLEO, my23VTM, SGsat24Bliss} have utilized stochastic geometry in performance evaluation for LEO satellite networks, only one recent study has focused on a GEO satellite constellation~\cite{my24TWC}. 

The link-level performance of hybrid satellite networks with GEO and LEO satellites has been analyzed \cite{wang2018novel, sharma2016line, gu2021dynamic, ryu2024rate}. A common and significant challenge in the hybrid satellite networks is \emph{in-line interference}, which occurs when a LEO satellite enters the line-of-sight (LOS) between a GEO satellite and its ground user. The earlier studies \cite{wang2018novel} and \cite{sharma2016line} introduced the adaptive power control techniques in a hybrid system comprising one GEO and one LEO satellite, which aims to maximize the system throughput by reducing the in-line interference. To address the collinear interference from multiple LEO satellites, a flexible spectrum sharing and cooperative transmission strategy were proposed in \cite{gu2021dynamic} for a hybrid system with one GEO satellite and multiple LEO satellites. In \cite{ryu2024rate}, a rate-splitting multiple access framework was applied to a hybrid GEO-LEO satellite system where a super-common message helps to mitigate the in-line interference. Although the aforementioned works \cite{wang2018novel, sharma2016line, gu2021dynamic, ryu2024rate} have competently addressed a wide range of aspects in hybrid GEO-LEO satellite systems, 
there have been few studies on the system-level performance of hybrid GEO-LEO satellite networks with large constellations.

Motivated by this, we aim to analyze downlink hybrid GEO-LEO satellite networks, where multiple GEO satellites coexist with LEO satellites and are randomly distributed along the geostationary orbit. The main contributions are described as follows.

\begin{itemize}
    \item \textbf{Unified modeling for hybrid GEO-LEO satellite networks: }
    GEO and LEO satellites exhibit distinct orbital characteristics. LEO constellations are typically designed using Walker Delta and Star patterns to ensure global coverage \cite{WalkerStar2019}. On the contrary, the positioning of GEO satellites is only restricted to the geostationary orbit along the equatorial plane, maintaining a fixed position relative to the Earth. Given the limited availability of orbital and frequency resources, cooperation between GEO and LEO satellites should be considered. In this context, we propose a unified model for hybrid GEO-LEO satellite networks that accounts for their respective orbital characteristics.
    
    \item  \textbf{Coverage analysis: } Based on the stochastic modeling, we first analyze satellite visibility and distance distributions, which vary based on the terminal's latitude. We then derive the association probabilities considering the bias factor for offloading. Finally, we establish the analytical expression of the coverage probability for hybrid GEO-LEO networks.
    
    
\end{itemize}

\section{System model}\label{sec:syst_model}
We consider a hybrid GEO-LEO satellite network 
where GEO satellites are distributed in the geostationary orbit, and the LEO satellites are located on a sphere.
Let $\mathbf{x}_n^{\idxGL}$, $\idxGL \in \{\G, \L\}$, denote the position of the $n$-th satellite for each type where the notation $\idxGL$ is used to distinguish between GEO and LEO satellites throughout the paper.
We assume that the satellite positions follow a homogeneous PPP $\BPP_{\idxGL}=\{\mathbf{x}_n^{\idxGL}\}$ with an intensity of $\lambda_{\idxGL}$.
Using spherical coordinates, the geostationary orbit is expressed as $\AG=\left\{\rho=\re+\aG, \psi=\frac{\pi}{2}, 0 \le\varphi\le 2\pi\right\}$, and the LEO sphere is $\AL=\left\{\rho=\re+\aL, 0\le\psi\le\pi, 0 \le\varphi\le 2\pi\right\}$ where $\rho$ is the radial distance, $\psi$ is the polar angle, and $\varphi$ is the azimuthal angle, $\re$ is the Earth's radius, and $a_{\idxGL}$ is the satellite altitude.
Our focus is on a typical terminal placed at arbitrary latitude $\phi$ and longitude~$\theta$ whose location is represented as
\begin{align}
\mathbf{t}=[\re \cos{\phi}\cos{\theta}, 
        \re \cos{\phi}\sin{\theta}, 
        \re \sin{\phi}]^{\T}.
        \end{align}
Since the terminal can only observe satellites above its horizontal plane, we define the \textit{visible regions} where the terminal can see GEO and LEO satellites as $\Asvis$, $\idxGL\in\{\G,\L\}$, as illustrated in Fig.~\ref{Fig:visible_region}. The regions  $\AGvis$ and $\ALvis$ are a geostationary orbit's arc and a LEO sphere's cap, respectively.
We also denote the region from which any distance to the terminal is less than~$r$ by $\mathcal{A}^{\idxGL}(r)$, $\idxGL \in \{\G,\L\}$.

For notational simplicity, the satellite index $n$ is arranged in order of distance for each type of satellites. In other words, $n=0$ indicates the satellite closest to the terminal, and as $n$ increases, it represents satellites that are progressively farther away. 
When a satellite is located at a position $\mathbf{x}_n^{\idxGL}$, the path-loss between the satellite and the terminal is given by
$\ell(\mathbf{x}_n^{\idxGL}) =  \left(\frac{c}{4\pi f_{\mathrm{c}}}\right)^2 \|\mathbf{x}_n^{\idxGL}-\mathbf{t}\|^{-\alpha}$
where 
$c$ is the speed of light, $f_{\mathrm{c}}$ is the carrier frequency, and $\alpha$ is the path-loss exponent.
Small-scale fading channels are modeled using Nakagami-$m$ fading, which accurately reflects the line-of-sight (LOS) characteristics of satellite channels \cite{my23WCL,my24TWC}. 
The cumulative distribution function (CDF) of the channel gain is given by $F_{h_n}(x)=1-e^{-m x}\sum_{q=0}^{m-1}\frac{(mx)^q}{q!}$.

We assume that the satellites utilize directional antennas to compensate for large path losses of satellite channels, and the terminal has an omnidirectional antenna. 
The effective antenna gain of a satellite channel is expressed as $G_n^\idxGL = \Gtn^\idxGL \Gr$
where $\Gtn^\idxGL$ is the transmit antenna gain of the satellite and $\Gr$ is the receive antenna gain of the terminal. 
For analytical tractability, it is assumed that the serving satellites directs its mainlobe toward the target terminals, while the beam of other satellites are relatively misaligned \cite{SGsatBPP20Otaki,SGsatPPP23Park,my24TCOMsub}.

\begin{figure}
\centering
\subfigure[visible arc of the geostationary orbit]{
\includegraphics[width=.65\columnwidth]{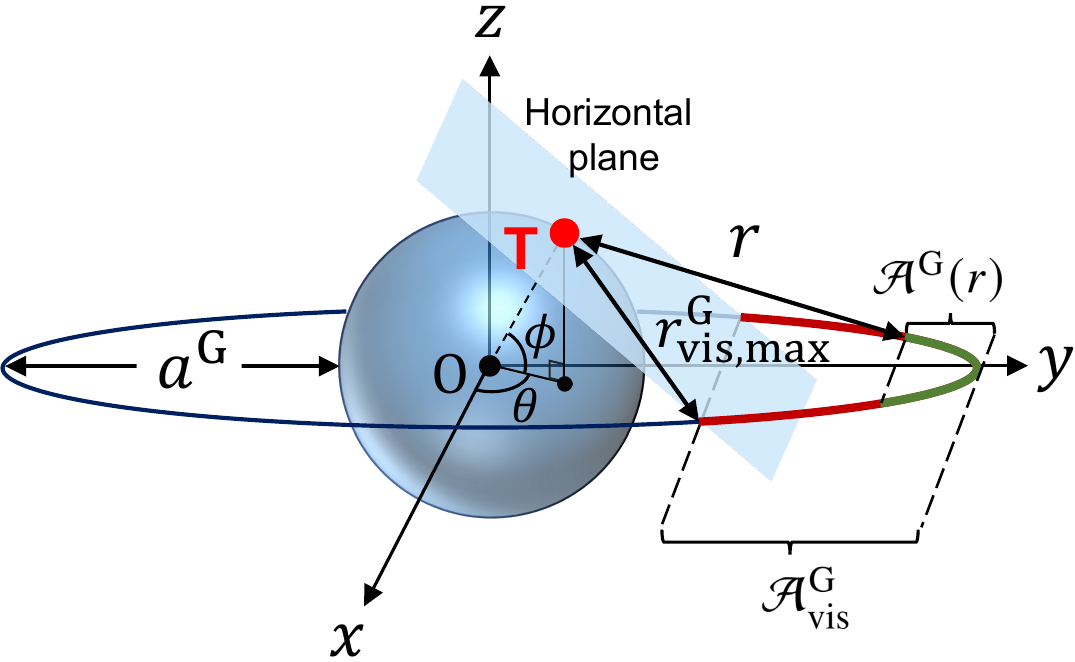}
\label{Fig:AGvis}
}
\subfigure[visible spherical cap of the LEO sphere]{
\includegraphics[width=.55\columnwidth]{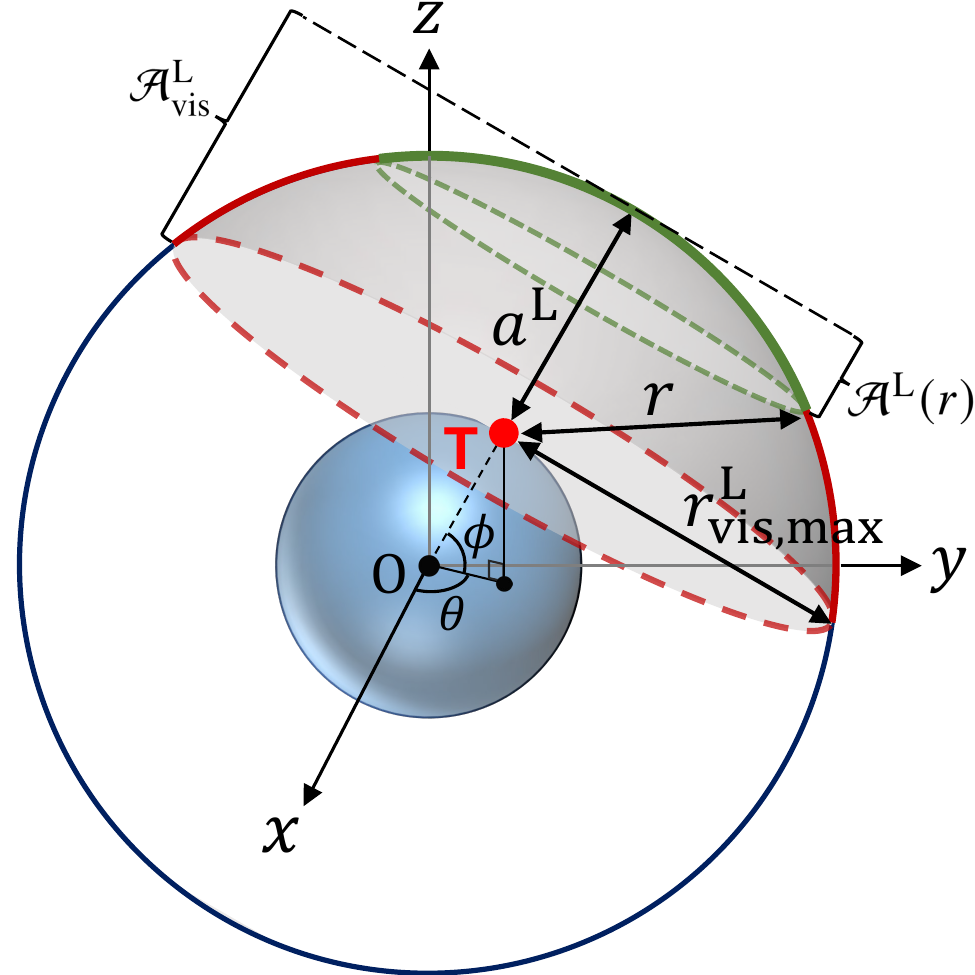}
\label{Fig:ALvis}
}
\caption{Satellite visible regions, i.e., $\Asvis$, $\idxGL\in\{\G,\L\}$.}
\label{Fig:visible_region}
\end{figure}


We adopt a biased-received-power (BRP)-based association policy where the typical terminal is associated to the satellite with the largest averaged BRP \cite{SG12HetNetJo}. 
Let $P_{\mathrm{r}}^{\idxGL}$ denote a long-term averaged BRP from the nearest GEO or LEO satellite. Then, we can express this as $P_{\mathrm{r}}^{\idxGL} = \Pt^{\idxGL} \Go^\idxGL \ell(\mathbf{x}_{0}^{\idxGL}) B^{\idxGL}$
where $\Pt^{\idxGL}$ is the transmit power and $B^{\idxGL}$ is a bias factor for load balancing between GEO and LEO satellites. For example, if $\BG \gg \BL$, the probability that the terminal is served by a GEO satellite increases, and vice versa. With the BRP-based association, the serving satellite is selected as $\asscidx = \text{argmax}_{\idxGL\in\{\G,\L\}} P_{\mathrm{r}}^{\idxGL}.$
After that, the received SINR at the terminal is given by
\begin{align}
\SINR_{\asscidx}=
        \frac{\Pt^{\asscidx} \Go^{\asscidx} h_0^{\asscidx} \ell(\mathbf{x}_{0}^{\asscidx})}{I_{\G | \asscidx} + I_{\L | \asscidx} + N_0 W}
\end{align}
where $I_{\G | \asscidx}$ and $I_{\L | \asscidx}$ are the aggregated interference from the GEO and LEO satellites given that the terminal is associated with a type of satellite $\asscidx \in \{\G,\L\}$, respectively.

\section{Mathematical preliminaries}

It is important to note that the terminal's latitude determines the visibility of GEO satellites, while the visibility of LEO satellites remains uniform in all directions. For example, when the terminal is positioned at the equator, i.e., $\phi=0$, the visible arc $\AGvis$, depicted as the red curve in Fig. \ref{Fig:AGvis}, reaches its maximum length. As the terminal moves farther from the equator and its latitude increases, the visible arc shortens and eventually vanishes when the terminal latitude exceeds $\phiinv = \cos^{-1}\left(\frac{\re}{\re + \aG}\right) \approx 81.3 \text{ degrees.}$ Based on this fact, the satellite-visible probabilities are given in the following lemma.

\begin{lem}
The probability that the terminal can see at least one GEO satellite is given by
\begin{align}
\Pvis^{\G} 
    =\begin{cases} 
            1-e^{2 \lG (\re+\aG)\cos^{-1}\left(\frac{\re}{(\re+\aG)\cos\phi}\right)}, & \mbox{if  } |\phi|<\phiinv,\\
            0, & \mbox{otherwise},
        \end{cases}
\end{align}
and the probability that the terminal can see at least one LEO satellite is  
\begin{align}
\Pvis^{\L}
    = 1-e^{-2 \pi \lL (\re+\aL)\aL}.
\end{align}
\end{lem}
The satellite-visible probabilities can be derived using the void probability of the PPP $\BPP_{\idxGL}$, i.e., the probability that no point exists in the visible region $\Asvis$, which is given by $e^{\lambda_{\idxGL}\Asvis}$. The details of deriving $\Pvis^{\L}$ and $\Pvis^{\G}$ are given in \cite{my22TCOM} and \cite{my24TWC}, respectively.

Let $R_0^{\idxGL}$ denote the distance
between the terminal and the closest visible GEO ($\idxGL=\G$) or LEO ($\idxGL=\L$) satellite.
When $|\phi|<\phiinv$, the minimum and the maximum distances to a visible GEO satellite are given by $\rminG=\sqrt{(\re+\aG-\re\cos\phi)^2 + \re^2 \sin^2\phi}$ 
and $\romaxG=\sqrt{\aG^2 + 2 \aG \re}$, respectively \cite{my24TWC}.
The minimum and maximum distances to a visible LEO satellite are given by $\aL$ and $\romaxL=\sqrt{\aL^2 + 2 \aL \re}$, respectively. Using these distances, we obtain the CDFs of $\RoG$ and $\RoL$ in the following lemma.
\begin{lem}\label{lem:CDFPDF-RoG}
The CDFs of $\RoG$ and $\RoL$ are respectively given by \cite{my22TCOM,my24TWC}
\begin{align}\label{eq:CDFRoG}
    F_{\RoG}(r) &= \frac{ 1-\Psi(r,\phi)}{ 1-\Psi(\romaxG,\phi)}, \,\,\, \mathrm{for}\,\,\, \rminG \le r < \romaxG,\\
    F_{\RoL}(r) &=
        \frac{ 1-\PsiL(r)}{ 1-\PsiL(\romaxL)}, \,\,\, \mathrm{for  }\,\,\, \aL \le r < \romaxL
\end{align}
where $\Psi(r,\phi)
        =\exp\left(-2 \lG (\re+\aG)\cos^{-1}\left(\frac{(\re+\aG)^2+\re^2-r^2}{2(\re+\aG)\re\cos\phi}\right)\right) $ and $\PsiL(r) = \exp\left({-\frac{\pi \lL (\re+\aL) (r^2-\aL^2)}{\re}}\right)$.
\end{lem}

According to the satellite visibility and the terminal's latitude, we can classify possible association cases as in Table~\ref{table:assc_case}. 
When only GEO satellites are visible, the terminal is certainly served by a GEO satellite, and vice versa. When both GEO and LEO satellites are visible, the terminal is associated with a satellite using the BRP-based association policy. 
By using the notation $\hat{x}^{\idxGL}$, $x\in\{\Pt,\Go, B, \alpha\}$, $\idxGL \in\{\G,\L\}$, we denote the ratio between system parameters for the two types of satellites as  $\hat{x}^{\G} \delequal x^{\G}/x^{\L} \delequal 1/\hat{x}^{\L}$. 
Then, the association probabilities are given in the following lemma.

\begin{table}[t]
\centering
\caption{Association cases based on the satellite visibility and the terminal's latitude}
\label{table:assc_case}
\begin{tabular}{|c|c|c|c|}
\hline
Latitude & GEO visibility & LEO visibility & Association \\ \hline
\multirow{4}{*}{\(|\phi| \le \phi_\text{inv}\)} 
 & Yes & Yes & GEO or LEO \\ \cline{2-4}
 & Yes & No & GEO \\ \cline{2-4}
 & No & Yes & LEO \\ \cline{2-4}
 & No & No & - \\ \hline
\multirow{2}{*}{\(|\phi| > \phi_\text{inv}\)}
 & No & Yes & LEO \\ \cline{2-4}
 & No & No & - \\ \hline
\end{tabular}
\end{table}

\begin{lem}\label{lem:Passc}
    When the terminal at a latitude less than $\phiinv$ sees at least one GEO and one LEO satellite, the probabilities that the terminal is associated with a serving GEO or LEO satellite are respectively given by
\begin{align}\label{eq:PasscG}
\PasscG&
= 1-\int_{r_{\mathrm{min}}^{\G}(\phi)}^{\romaxG} F_{\RoL} (\roLbiased(r)) f_{\RoG}(r) dr,\\
\PasscL&
= 1-\int_{\aL}^{\romaxL} F_{\RoG} (\roGbiased(r))  f_{\RoL}(r) dr,
\end{align}
where 
$d^{\idxGL}(r) \delequal (\Pthat^{\idxGL} \Gohat^{\idxGL} \Bhat^{\idxGL} )^{1/\alpha^{\idxGL}} r^{1/\alphahat_{\idxGL}}$, $\idxGL\in\{\G,\L\}$, is the biased minimum distance to the nearest satellite.
\end{lem}
\begin{proof}
    The probability that the terminal is associated with a GEO satellite is derived as
    \begin{align}
    \PasscG 
    &= \P\left[\PrG > \PrL \right]\nonumber\\
    &= \P\left[\Pt^{\G} \Go^\G \ell(\mathbf{x}_{0}^{\G}) B^{\G} > \Pt^{\L} \Go^\L \ell(\mathbf{x}_{0}^{\L}) B^{\L} \right]\nonumber\\
    &= \!\int_{r_{\mathrm{min}}^{\G}(\phi)}^{\romaxG} \! \P\left[\RoL > (\Pthat^{\L} \Gohat^{\L} \Bhat^{\L} )^{1/\alphaL} r^{1/\alphahat_{\L}}\right] f_{\RoG}(r) dr 
    \end{align}
    Herein, $f_{\RoG}(r)$ is the probability density function of $\RoG$, which can be readily obtained by differentiating the CDF \eqref{eq:CDFRoG}.    
    Using the CDF of $\RoL$ and the definition of $d^{\idxGL}(r)$, we can obtain \eqref{eq:PasscG}.
    The association probability $\PasscL$ can be similarly obtained but omitted due to space limitation.
\end{proof}

\begin{figure*}
\begin{center}
\includegraphics[width=1.9\columnwidth]{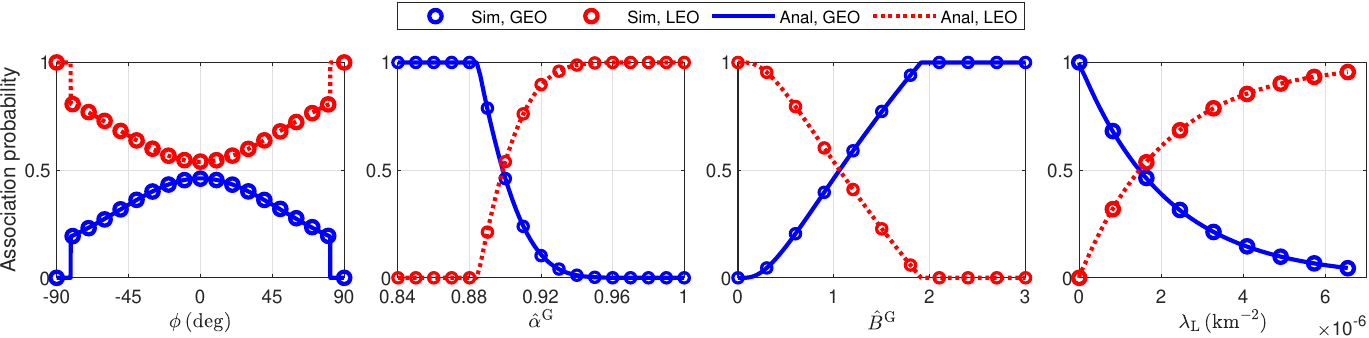}
\end{center}
\setlength\abovecaptionskip{.25ex plus .125ex minus .125ex}
\setlength\belowcaptionskip{.25ex plus .125ex minus .125ex}
\vspace{-8pt}
\caption{Association probabilities. Unless otherwise stated, we set $\{\alphaG,\alphaL\}=\{3.6, 4\}$, $\Pthat^{\G} \GohatG=50$ dB, and $m=3$.}
\label{Fig:Passc}
\end{figure*}

\begin{figure}
\begin{center}
\includegraphics[width=.9\columnwidth]{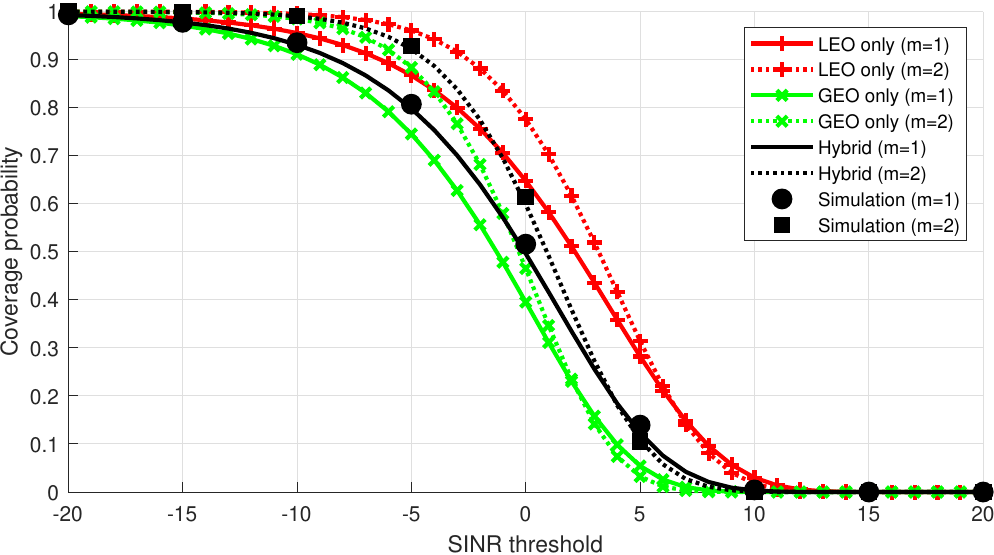}
\end{center}
\setlength\abovecaptionskip{.25ex plus .125ex minus .125ex}
\setlength\belowcaptionskip{.25ex plus .125ex minus .125ex}
\caption{Coverage probability versus SINR threshold $\tau$ with various Nakagami fading parameters $m=\{1,2\}$ where $\{\alphaG,\alphaL\}=\{2.7, 3\}$.}
\label{Fig:Pcov_vs_tau}
\end{figure}


\section{Coverage analysis}\label{sec:Pcov}
Driving the Laplace transform of aggregated interference is an important intermediate step for stochastic geometry-based coverage analysis. Thus, we first obtain the Laplace transform of the aggregated interference $I_{\idxGL}$, $\idxGL \in \{\G,\L\}$, given $R_0^{\idxGL}=r_0$ in the following lemma.

\begin{lem}\label{lem:LT}
    When the effective antenna gains of the interfering satellites are uniform, i.e., $\Gn^{\idxGL}=\bar{G}^{\idxGL}\ \forall n \neq 0$, $\idxGL\in\{\G,\L\}$, for some constant $\bar{G}^{\idxGL}$, the Laplace transforms of the total interference power from GEO and LEO satellites are derived as follows: 
    \begin{align}\label{eq:LTG}
    {\mathcal{L}}_{I_{\G | \idxGL}}(s;r_0) 
        &=e^{- 4 \lG (\re+\aG) \Omega(s;r_0)},\\
        \mathcal{L}_{I_{\L | \idxGL}}(s;r_0) 
         &=e^ {- 2 \pi \lL \frac{\re+\aL}{\re}\kappa(s; r_0)}
    \end{align}  
    where 
    $\Omega(s;r_0) = \int_{r_0}^{\romaxG} \left(1-\left(\frac{m \coeffG r^{\alpha}}{s + m \coeffG r^{\alpha}}\right)^{m}\right) \frac{r dr}{\sqrt{v_1-\left(v_2-r^2\right)^2}}$
    and  $\kappa(s; r_0)= \int_{r_0}^{\romaxL} \left( 1 - \left( \frac{m \coeffL r^{\alphaL}}{s+m \coeffL r^{\alphaL}} \right)^m \right)  r dr$ with $v_1=4(\re+\aG)^2\re^2\cos^2\phi$, $v_2=(\re+\aG)^2+\re^2$, and $\coeff^{\idxGL} = \frac{16 \pi^2\fc^2}{\Pt^{\idxGL} \bar{G}^{\idxGL} c^2}$.
\end{lem}
    For more details, see Appendix \ref{App:LTG}.\\
\vspace{-5pt}
We define a satellite as providing coverage when the $\SINR$ at a typical terminal meets or exceeds a specified threshold $\tau$. The coverage probability is subsequently expressed as $\P[\SINR\geq \tau]$. Based on the result from Lemma~\ref{lem:LT}, the coverage probability is formulated in the following theorem.

\begin{thm}\label{thm:Pcov_tot}
 The total coverage probability for hybrid GEO-LEO satellite networks is given by
    \begin{align} \label{eq:Pcov}
    \Pcov
    & =\Pvis^{\G} \Pvis^{\L} (\PasscG \PcovG + \PasscL \PcovL ) \notag\\
    &\quad + \Pvis^{\G} (1-\Pvis^{\L}) \PcovGhat + (1-\Pvis^{\G}) \Pvis^{\L} \PcovLhat
    \end{align}
    In \eqref{eq:Pcov}, $\PcovG$ and $\PcovL$ are the coverage probabilities when the terminal is associated with either a GEO or LEO satellite, respectively, which are given by
    \begin{align}\label{eq:PcovG}
        \PcovG 
        &=\sum_{i=1}^{m} \binom{m}{i}(-1)^{i+1} \int_{\rminG}^{\romaxG} e^{-i  N_0 W \delta_{\G}} \notag\\
        & \times  \mathcal{L}_{I_{\G|\G}}(i \delta_{\G};r) \mathcal{L}_{I_{\L|\G}}(i \delta_{\G};\roLbiased(r)) f_{\RoG}(r)dr
    \end{align}
    and 
    \begin{align}\label{eq:PcovL}
        \PcovL      
        &= \sum_{i=1}^{m} \binom{m}{i}(-1)^{i+1} \int_{\aL}^{\romaxL} e^{-i N_0 W \delta_{\L}}\notag\\
        &\times   \mathcal{L}_{I_{\G|\L}}( i \delta_{\L};\roGbiased(r)) \mathcal{L}_{I_{\L|\L}}(i \delta_{\L};r) f_{\RoL}(r)dr
    \end{align}
    where $\nu = m(m!)^{-1/m}$ and $\delta_{\idxGL} = \nu \tau \coeff^{\idxGL}_0 r^{\alpha_{\idxGL}}$ with $\coeff^{\idxGL}_0 = \frac{16 \pi^2\fc^2}{\Pt^{\idxGL} G^{\idxGL}_0 c^2}$.
    Herein, $\PcovGhat$ and $\PcovLhat$ are the coverage probabilities without interference from the other type of satellites, which can be obtained by removing $\mathcal{L}_{I_{\L|\G}}(i \delta_{\G};\roLbiased(r))$ and $\mathcal{L}_{I_{\G|\L}}( i \delta_{\L};\roGbiased(r))$ in \eqref{eq:PcovG} and \eqref{eq:PcovL}, respectively.
\end{thm}
For more details, see Appendix \ref{App:Pcov_tot}.


\section{Simulation results}\label{sec:sim_res}





In this section, we numerically verify the derived results with simulations. 
We let $N_\idxGL$ denote the average number of satellites, which determines satellite densities, i.e.,
$\{\lambda_\G, \lambda_\L\}=\left\{\frac{\NG}{2\pi(\re+\aG)}, \frac{\NL}{4\pi(\re+\aL)^2}\right\}$. Unless otherwise stated, we set $\re=6378$ km, $c=3\times10^5$ km/s, $\{\aL, \aG\}=\{600, 35786\}$ km, $\{\NL,\NG\}=\{100,1000\}$, $\{\phi,\theta\}=\{0,0\}$ deg, $\Bhat^{\G}=1$, and $N_0=-174$ dBm/Hz. 
We use the satellite and handheld terminal parameters in the Ka-band scenario of the 3GPP NTN standard \cite{TR38.821}: $\fc=20$ GHz, $W=30$ MHz, the maximum effective isotropically radiated power (EIRP) density for GEO satellites $=40$ dBW/MHz, and that for LEO satellites $=4$ dBW/MHz. From these EIRP densities, the transmit power of the satellites can be obtained using the fact that the EIRP density is equal to
$\Pt^\idxGL \Go^\idxGL/W$.
The analytical results of the association probabilities are from Lemma \ref{lem:Passc}, and those of the coverage probability are derived from Theorem 1.
The simulation results are obtained through Monte Carlo simulations with 100k independent realizations.

Fig. \ref{Fig:Passc} verifies that the analytical results of the association probabilities in Lemma \ref{lem:Passc} are in strong accordance with the simulation results. The association probability for a GEO satellite decreases with $|\phi|$ and then vanishes at $\phi\approx \pm 81.3$ deg. This is because the GEO satellites, distributed on the equatorial plane, hardly provide their coverage to the areas close to the polar regions. Thus, in order to achieve global coverage, it is essential to deploy LEO satellites in polar orbits, e.g., Walker Star constellations. The association probability with GEO satellites decreases as $\alphahat^{\G}$ increases because the severe path loss of GEO satellite links discourages associating with GEO satellites. Moreover, the association probability with GEO satellites increases with $\Bhat^{\G}$, implying that biasing factors can be effectively used to offload data between GEO and LEO satellite networks. When a limited number of LEO satellites are deployed, i.e., for a small $\lambda_\L$, GEO satellites are more likely to serve terminals, and vice versa.

Fig.~\ref{Fig:Pcov_vs_tau} shows the coverage probability versus the $\SINR$ threshold $\tau$. We compare the performance of hybrid GEO-LEO networks to that of satellite networks where only GEO or LEO satellites exist. 
As expected, the coverage probabilities decrease with the $\SINR$ threshold due to the definition $\Pcov=\P[\SINR \ge \tau]$. The analytical results obtained from Theorem 1 exhibit performance that aligns well with the simulation results for various Nakagami parameters $\alpha=\{1,2\}$. It is shown that  the cooperation between GEO and LEO satellites extends the coverage of overall satellite networks, while the load on LEO satellite networks can be offloaded to GEO satellites at the expense of reduced coverage.

\section{Conclusions}\label{sec:conclusions}
In this paper, we considered a hybrid GEO-LEO satellite network, with GEO satellites distributed along the equatorial plane and LEO satellites distributed on the sphere. We first analyzed satellite-visible probabilities, distance distributions, and association probabilities. Then, we derived an analytical expression for the network's coverage probability. Using Monte Carlo simulations, we validated the analytical results and demonstrated how system parameters affected coverage performance. The analytical results can provide an effective estimate of coverage performance in scenarios where GEO and LEO satellites cooperate or share the same frequency bands.
We believe this study can be readily extended to scenarios with deterministic GEO satellites or restricted arcs of the geostationary orbit.

\appendices

\section{Proof of Lemma \ref{lem:LT}}\label{App:LTG}
Let $\mathcal{X}_{\idxGL}$ denote the region where interfering satellites can be located given $R_0^{\idxGL}=r_0$, i.e., $\mathcal{X}_{\idxGL}=\Asvis\cap\As(r_0)^{\mathrm{c}}$. 
Then, the Laplace transform of the aggregated GEO interference given $\RoG=r_0$ is derived as 
\begin{align}\label{eq:LI_approx_proof}
    &\mathcal{L}_{I_{\G|\idxGL}}(s;r_0)
        = \E\left[e^{-s I_{\G|\idxGL}}|\RoG = r_0 \right]\nonumber\\
        &= \E_{\BPPG,\{h_n^{\G}\}}\left[\exp\left(-s\sum_{n=1}^{{\BPPG}(\mathcal{X}_{\G})} \PtG \bar{G}^{\G}  h_n^{\G} \ell(\mathbf{x}_n^{\G})\right)\right]\nonumber\\
        &\mathop=\limits^{(a)} \exp \biggl(-\lG \!\!\int_{\mathbf{x} \in \mathcal{X}_{\G}} \!\!\left( 1\!-\!\E_{h_n^{\G}}\!\!\left[e^{-s \Pt \bar{G}^{\G} h_n^{\G} \ell(\mathbf{x})}\right] \right) d\mathbf{x} \bigg) \nonumber\\
        &\mathop=\limits^{(b)} \exp \Biggl(-\lG \!\!\int_{\mathbf{x} \in \mathcal{X}_{\G}} \left( 1- \left(\frac{m \coeffG r^{\alpha}}{s + m \coeffG r^{\alpha}}\right)^{m} \right) d\mathbf{x} \Biggr)
    \end{align}
    where ($a$) follows from Campbell’s theorem, and ($b$) follows from the Laplace transform $\E[e^{-s h_n}]=\mathcal{L}_{h_n}(s)=\left(\frac{m}{s+m}\right)^m$. Using the fact that $\frac{d|\AG(r)|}{dr}
        =\frac{4r(\re+\aG)}{\sqrt{v_1-\left(v_2-r^2\right)^2}}$, 
     we can obtain \eqref{eq:LTG}. Similarly, the Laplace transform of the aggregated LEO interference given $\RoL=r_0$, $\mathcal{L}_{I_{\L|\idxGL}}(s;r_0)$, can be derived 
     using the fact that $\frac{d|\mathcal{\AL}(r)|}{dr}=\frac{2\pi r(\re+\aL)}{\re}$. The details of the derivation are omitted due to space limitations.

\section{Proof of Theorem \ref{thm:Pcov_tot}}\label{App:Pcov_tot}
There are four visibility cases, as described in Table \ref{table:assc_case}. The probability that both GEO and LEO satellites are visible is given by $\Pvis^{\G}\Pvis^{\L}$, and the coverage probability for this case is derived as $\PasscG \PcovG + \PasscL \PcovL$ where $\Pcov^{\asscidx} = \P[\SINR_{\asscidx} \geq \tau]$. 
Similarly, the probabilities that only GEO or LEO satellites are visible are given by $\Pvis^{\G}(1-\Pvis^{\L})$ and $(1-\Pvis^{\G})\Pvis^{\L}$, respectively, and their corresponding coverage probabilities are $\bar{\mathcal{P}}_{\text{cov}}^{\G}$ and $\bar{\mathcal{P}}_{\text{cov}}^{\L}$, respectively.
Thus, using the total probability theorem, the total coverage probability can be expressed as~\eqref{eq:Pcov}.
    The coverage probability $\PcovG$ can be derived as
    \begin{align}\label{eq:PcovG-1}
    &\PcovG =\E_{\RoG}\left[ \P\left[\frac{(1/\coeffG_0) h_0^{\G} r^{-\alphaG}}{\IGG + \ILG + N_0 W} \geq \tau \,\bigg|\,\asscidx=\G, \RoG=r \right]\right]\nonumber\\
    &= \int_{\rminG}^{\romaxG} \P\left[h_0^{\G} \geq  (\IGG + \ILG + N_0 W) \tau \coeffG_0 r^{\alphaG}  \right]f_{\RoG}(r)dr
    \end{align}
    To further simplify \eqref{eq:PcovG-1}, we use the approximated CDF of the channel gain $F_{h_n}(x) \approx 1-\sum_{i=1}^{m}\binom{m}{i}(-1)^{i+1}e^{-\nu i x}$ \cite{SG17mmWave}. Then, the probability in \eqref{eq:PcovG-1} becomes    
    \begin{align}\label{eq:PcovG-1-prob}
         &\E_{\IGG, \ILG}\left[\sum_{i=1}^{m}\binom{m}{i}(-1)^{i+1}e^{- i  (\IGG + \ILG + N_0 W) \delta_{\G}} \right]\notag\\
         & \mathop = \limits^{(a)} \sum_{i=1}^{m} \binom{m}{i}(-1)^{i+1} e^{-i  N_0 W \delta_{\G}} \mathcal{L}_{\IGG}(i \delta_{\G};r) \mathcal{L}_{\ILG}(i \delta_{\G};\roLbiased(r))
    \end{align}
    where ($a$) follows from the independence between $\IGG$ and $\ILG$ and the definition of the Laplace transform. 
    From \eqref{eq:PcovG-1} and \eqref{eq:PcovG-1-prob}, the derivation of $\PcovG$ is complete.
    The coverage probability without interference from LEO satellites, i.e., $\PcovGhat$, can be readily obtained by removing $\mathcal{L}_{\ILG}$ from $\PcovG$.    
    Similarly, $\PcovL$ and $\PcovLhat$ can be obtained but the derivation is omitted due to space limitation.

\ifCLASSOPTIONcaptionsoff
  \newpage
\fi

\bibliographystyle{IEEEtran}
\bibliography{
    references/3GPP,
    references/books,
    references/chSR,
    references/IEEEabrv, 
    references/myPapers,
    references/refs,
    references/SG,
    references/hybrid_GEO+LEO
    }

\end{document}